\documentclass[10pt,twocolumn]{IEEEtran}

\usepackage{rotate,epsf,epsfig,url,subfigure,wrapfig,algorithm2e}
\usepackage[sc]{mathpazo}
\usepackage{amsfonts,amssymb,graphicx,cite}
\usepackage{psfrag}
\usepackage{amsthm}
\usepackage[usenames,dvipsnames]{color}
\usepackage{palatino}
\usepackage{epsfig}
\usepackage{amsmath}
\usepackage{latexsym}
\usepackage{amsfonts}
\usepackage{amssymb}
\usepackage{calrsfs}
\usepackage{enumitem}

\usepackage{url}
\usepackage{subfigure,nicefrac}
\usepackage{wrapfig}
\usepackage{wasysym}   

\usepackage{soul}

\makeatother
\newcommand{\va}{\mbox{${\bf a}$}}

\newcommand{\vxt}{\mbox{${\bf \tilde x}$}}
\newcommand{\vh}{\mbox{${\bf h}$}}

\newcommand{\ve}{\mbox{${\bf e}$}}

\newcommand{\vw}{\mbox{${\bf w}$}}

\newcommand{\vy}{\mbox{${\bf y}$}}
\newcommand{\vz}{\mbox{${\bf z}$}}

\newcommand{\mE}{\hbox{{\bf E}}}

\newcommand{\mH}{\hbox{{\bf H}}}

\newcommand{\mHt}{\mbox{${\bf \tilde H}$}}
\newcommand{\mI}{\hbox{{\bf I}}}

\newcommand{\mQ}{\hbox{{\bf Q}}}
\newcommand{\mR}{\mbox{{$\bf R$}}}

\newcommand{\mS}{\mbox{{$\bf S$}}}

\newcommand{\mW}{\hbox{{\bf W}}}

\newcommand{\gd}{\delta}
\newcommand{\gre}{\varepsilon}

\newcommand{\gth}{\theta}

\newcommand{\gl}{\lambda}
\newcommand{\gm}{\mu}

\newcommand{\gs}{\sigma}

\newtheorem{theorem}{Theorem}
\newtheorem{lemma}[theorem]{Lemma}
\newtheorem{prop}{Proposition}
\newtheorem{claim}[theorem]{Claim}

\newtheorem{definition}{Definition}
\theoremstyle{remark}
\newtheorem{question}{Question}
\newtheorem{coro}{Corollary}

\newcommand{\beq}{\begin{equation}}
\newcommand{\eeq}{\end{equation}}
\newcommand{\bea}{\begin{array}}
\newcommand{\ena}{\end{array}}
\newcommand{\bds}{\begin {description}}
\newcommand{\eds}{\end {description}}
\newcommand{\bdf}{\begin{definition}}
\newcommand{\blm}{\begin{lemma}}
\newcommand{\edf}{\end{definition}}
\newcommand{\elm}{\end{lemma}}
\newcommand{\bthm}{\begin{theorem}}
\newcommand{\ethm}{\end{theorem}}
\newcommand{\bprp}{\begin{prop}}
\newcommand{\eprp}{\end{prop}}
\newcommand{\bcl}{\begin{claim}}
\newcommand{\ecl}{\end{claim}}
\newcommand{\bcr}{\begin{coro}}
\newcommand{\ecr}{\end{coro}}
\newcommand{\bquest}{\begin{question}}
\newcommand{\equest}{\end{question}}

\begin{document}
\title{Eliminating  Interference in LOS Massive Multi-User MIMO with a Few Transceivers}
\author{Uri Erez\thanks{U. Erez is with Tel Aviv University, Tel Aviv,
Israel (email: uri@eng.tau.ac.il). The work of U. Erez was supported by by the ISF under Grant 1956/15.} and Amir Leshem\thanks{Amir Leshem is with Faculty of Engineering, Bar-Ilan University (leshema@biu.ac.il). The work of Amir Leshem was partially support by ISF grant 1644/18 and ISF-NRF grant  2277/16.}}
\maketitle



\begin{abstract}
Wireless cellular communication networks are bandwidth and interference limited. An important means to overcome these resource limitations is the use of multiple antennas. Base stations equipped with a very large (massive) number of antennas have been the focus of recent research.
A bottleneck in such systems is the limited number of transmit/receive chains. In this work, a line-of-sight (LOS) channel  model is considered. It is shown that for a given number of interferers, it suffices that the number of transmit/receive chains exceeds the number of desired users by one, assuming a sufficiently large antenna array.  From a theoretical point of view, this is the first result proving the near-optimal performance of antenna selection, even when the total number of signals (desired and interfering) is larger than the number of receive chains.
Specifically, a single additional chain suffices to reduce the interference to any desired level. We prove that using the proposed selection, a simple linear receiver/transmitter for the uplink/downlink provides near-optimal rates. In particular, in the downlink direction, there is no need for complicated dirty paper coding; each user can use an optimal code for a single user interference-free channel. In the uplink direction, there is almost no gain in implementing joint decoding.
The proposed approach is also a significant improvement both from system and computational perspectives. Simulation results demonstrating the performance of the proposed method are provided.
\end{abstract}

\section{Introduction}

It is well known that given an adaptive array with $N_r$ receive chains, one can null out $N_r-1$ (single-antenna) interferers and enjoy a full degree-of-freedom (DoF) for one desired source.
Such an architecture, often referred to as receive beamforming, is attractive due to its ease of implementation, as well as robustness.
Nonetheless, it leads  to low spectral utilization.



A line of work that has emerged in recent years, starting with  the seminal work of \cite{marzetta2010noncooperative}, is the use of massive multi-input multi-output (MIMO) technology as a promising practical means to boost the throughput of wireless networks. In a nutshell, assuming a rich scattering environment or alternatively that the number of antennas at the base station is much larger than the  total number of users, one can leverage channel hardening arguments to conclude that the channel vectors of the users will be approximately orthogonal. Under these conditions, beamforming will be near optimal on both uplink and downlink.

Nonetheless, massive MIMO technology faces some considerable challenges, including the overhead incurred by the need for transmission of many pilots, and the substantial cost of radio frequency (RF) chains.
Specifically, we consider the class of line-of-sight wireless channels
\cite{almers2007survey}
which are very relevant in recent applications of mm wave wireless communications, see e.g., \cite{rappaport2013millimeter} and references therein.
In the present paper, we demonstrate that one can attain the performance gains of massive MIMO technology while employing a minimal number of receive chains.

Our main result is as follows: For any number of interferers, given $r$ receive chains and a large enough number $N_r$ of antennas, one can approximately  null out all interferers while simultaneously affording (with probability one) a full DoF to $r-1$ single-antenna sources. This corresponds to a utilization of $1-\nicefrac{1}{r}$ of the $r$ DoFs available.
In a practical system implementation, employing a linear massive MIMO array (assuming planar geometry), this may be accomplished by judiciously selecting $r$ antennas out of a large array of $N_r$ antennas. The output of the selected antennas is then fed into $r$ receive chains as depicted in Figure~\ref{fig:selection}.
\begin{figure}[htp]
\begin{center}
\includegraphics[width=0.9\columnwidth]{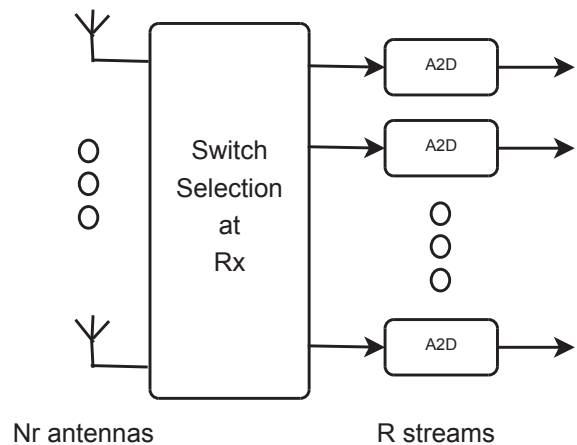}
\end{center}
\caption{Setting antenna spacing via selection.}
\label{fig:selection}
\end{figure}


\subsection{Related Work}
Massive multi-user MIMO (MU-MIMO) emerged from the realization that when the number of antenna elements is much larger than the number of users, the inner product between different users' channels becomes negligible with respect to the norm of each user's channel, leading to ``channel hardening''. Under such conditions,  linear processing in conjunction with single-user encoding/decoding (for the downlink/uplink respectively)  is near optimal \cite{marzetta2010noncooperative,rusek2012scaling,larsson2014massive}.

Digitally sampling a very large antenna array entails a  prohibitively large hardware complexity. Hence, substantial research efforts have been devoted to alleviating this burden. Among the proposed techniques are hybrid analog/digital beamforming, low resolution ADCs and antenna selection.
An overview of these various techniques can be found in, e.g.,
\cite{liang2014low,mollen2017uplink,mehanna2013joint}.

One may think of the selection mechanism as a means of altering the physical communication channel to attain some benefit.
The idea of of changing the physical propagation channel bears some relation to ``media-based modulation", ``spatial modulation" and ``index modulation" schemes; see \cite{basar2017index,ishikawa201850} for an overview of these inter-related concepts.
In all of these works, the physical medium is modulated based on the information-bearing signal.  In contrast, in the present work the physical medium is altered depending only on the location of the desired and interference sources, independent of the transmitted data.

\begin{figure*}
    \centering
    \includegraphics[width=0.85 \textwidth]{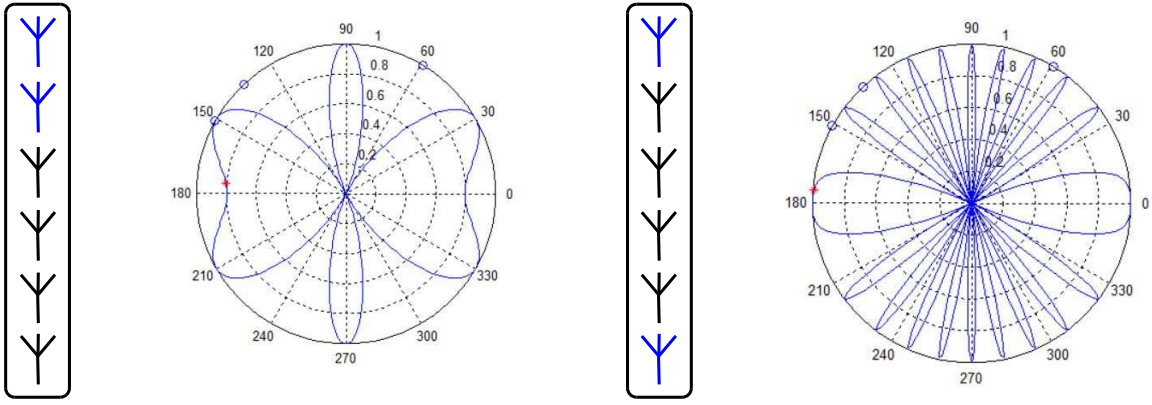} 
    \caption{Effect of selection on beam pattern.}
     \label{fig:ambiguous}
\end{figure*}




\section{System Model}

For simplicity, we consider the uplink portion of the communication link in a cellular MU-MIMO scenario.
We begin by defining the system configuration of a single link between a single (possibly MIMO) transmitter and the base station.

Denote by $t/N_t/N_r/r$  a transmit/receive MIMO configuration utilizing antenna selection.
Thus, the uplink of a massive MIMO cell with a total of $K$ single-antenna users will consist of $K$ links of the form $1/1/N_r/r$  as depicted in Figure~\ref{fig:selection}.
.

\section{Line-of-Sight Channel Model}

We consider the LOS channel model. This channel model is important in its own right, e.g., for millimeter-wave communication ~\cite{rappaport2015wideband}. Moreover, this channel model can serve as  the basis for a more elaborate, specular multipath MU-MIMO channel model.

Denoting  the channel vector  from transmitter $k$ to the $N_r$ antennas by $\vh_{k}\in \mathbb{C}^{N_r\times 1}$, the  received signal is given by
\begin{align}
\label{eq:MU_MIMO_model}
    \vy=\sum_{k=1}^K \mathbf{h}_{k} \mathbf{x}_k + \vz
\end{align}
where $\mathbf{z}$ is i.i.d.  circularly-symmetric complex Gaussian noise with variance $\gs^2$. The receiver decodes the messages of users $1\,\ldots,r-1$ whereas the signals $x_k$, $k=r,\ldots,K$
are interference.

We note that the $1/1/N_r/r$ MU-MIMO interference channel is equivalent to requiring that  the receiver must employ a linear front-end selection matrix $\mS_R \in \{0,1\}^{N_r \times N_r} $,  having exactly $r$ non-zero elements which are not in the same row or column. Applying selection matrix $\mS_{R}$   \eqref{eq:MU_MIMO_model} becomes
\begin{align}
\label{eq:IC_model2c}
    \mathbf{y}&=\mS_{R}^H \sum_{k=1}^K  \mathbf{h}_{k} \mathbf{x}_k + \mS_{R}^H \vz.
\end{align}
From a practical perspective, implementing the selection mechanism yields a substantial reduction in hardware complexity.

We consider the  LOS  channels. Specifically, we make the following assumptions:
\begin{enumerate}[label=A\arabic*]
  \item For simplicity we 
  assume planar geometry where all sources are far field point sources.
  \item The vectors $\vh_k$  consist of array manifold vectors
  \begin{align}
\vh(\theta_k)=\left[h_1(\theta_k),...,h_{N_r}(\theta_k)\right]^T \in {\mathbb C^{N_r \times 1}},
\label{def:h}
\end{align}
where
\begin{align}
h_n(\theta_k)= e^{2 \pi j \cdot n \cos(\theta_k)}, n=1,\ldots,N_r.
   \label{def:h2}
\end{align}
\item For sake of analysis, the dimensions of the array are taken to be unbounded similarly to \cite{marzetta2010noncooperative}.
  \item We  assume that the receiver has perfect CSI w.r.t. all channel gains corresponding to impinging signals. Transmitters on the other hand need not have access to any CSI beyond the rate at which they should communicate.
  \item Without loss of generality, we use the array manifold as the channel, since the signal attenuation can be absorbed in the power of $x_k$.
  \item \label{independence} For simulations we assume that the  locations  of the sources are independently  uniformly distributed in angle with respect to the receiver.
  \item We assume that the transmit power of all transmitters is bounded by $P$.
\end{enumerate}
Under these assumptions for the $1/1/N_r/r$ MU-MIMO channel equation~\eqref{eq:IC_model2c} reduces to
\begin{align}
\label{eq:IC_model2_cellular}
    \mathbf{y}=\mS_R^H \sum_{k=1}^K \vh(\theta_k) \mathbf{x}_k + \vz.
\end{align}

In the next section we prove that there is a selection such that with a proper linear processing  we only receive the desired $r-1$ signals, and attenuate the interfering signals to any prescribed level, with negligible noise amplification.

\section{Ergodic Beamforming}
Selection methods for reducing the complexity of massive MU-MIMO systems is a widely explored field \cite{mehanna2013joint, gershman2005space, sadek2007active, gorokhov2003receive}. There are many works related to antenna selection in LOS channels, see e.g., \cite{gao2013antenna, dong2011adaptive}. The optimal selection of antennas is very complicated. Therefore, many innovative techniques for antenna selection  have been proposed \cite{mehanna2013joint, gao2013antenna, gorokhov2003receive, gharavi2004fast}. Our contribution is establishing fundamental limits for antenna selection in LOS channels.  Namely, it is shown that, any number of interferers can be suppressed with the number of receive chains exceeding the number of desired users by one.

Classical signal processing literature is focused on Nyquist beamformers, where at least some antennas are separated by at most $\gl/2$. In this case, the array has a single main lobe at the desired direction, and the resolution of the array is determined  by the farthermost elements. The reason for this is that when all distances between antennas are larger than   $\gl/2$, an ambiguous beam pattern occurs. An example of this phenomenon is depicted in Figure~\ref{fig:ambiguous}.
Interestingly, ambiguous arrays can prove extremely advantageous when dealing with interference since such arrays have multiple nulls. We will show that by judiciously designing the beam pattern, we can simultaneously point multiple nulls at the multiple interferers.


Using a beamforming vector $\vw$, the received signal becomes:
\begin{equation}
\label{eq:rx_bf}
    \vxt=\vw^H \left(\sum_{j=1}^K \vh(\gth_k)  \mathbf{x}_k + \vz \right).
\end{equation}
The support of $\vw_k, k=1,\ldots,K$ implicitly defines the selection matrix $\mS_R$.
For each desired user $1 \le k \le r-1$ define the beamformer $\vw_{n_k}$ by
  \begin{align}
     \vw_{n_k}=\frac{1}{\sqrt{2}}\left(\ve_0+\ve_{n_i}\right)
     \label{def:vw}
 \end{align}
 where $\ve_n\in {\mathbb C}^{N_r}\times 1$ are the standard unit vectors.
Denote
\begin{align}
     g(\gth;\vw_{n})&=\left|\vw^H_{n} \vh(\theta) \right|^2 \\
     &=1+\cos(2\pi n\cos(\theta))
\label{def:g}
\end{align}
where $\vh(\theta)$ is defined in \eqref{def:h2}.
 We now show that  by selecting a common reference antenna, we can choose one antenna ($n_k$ ) per desired source ($1 \le k \le r-1 $ ) such that the beamformer
  satisfies:
\begin{equation}
\label{eq:design_eq}
 g(\gth_j;\vw_{n_k}) \approx \left\{
\begin{array}{cc}
 1    &  k=j \\
0     &  {\rm otherwise}
\end{array} \right.
\end{equation}
where without loss of generality the desired users are $k=1,...,r-1$.

\begin{theorem}[Receiver for a desired direction \cite{leshem2019ergodic}]
\label{thm:main}
Assume that the directions $\gth_1,..,\gth_K$ are such that $\cos(\gth_1),...,\cos(\gth_K)$ are independent over $\mathbb Q$. Then for every $i$ and every $\gd>0$, one can find $n_i \in \mathbb N$ such that beamforming with the vector $\vw_{n_i}$  yields:

\beq
 \label{eq:array_gain}
\begin{array}{lcll}
g(\gth_k;\vw_{n_i})&<& \gd,  \quad \quad k \neq i\\
g(\gth_i;\vw_{n_i})&>&1-\gd.
\end{array}
\eeq
\end{theorem}
The proof is based on the uniform distribution property of sequences modulo $1$ \cite{weyl1916gleichverteilung}. The proof can be found in  \cite{leshem2019ergodic}.

A simple corollary of Theorem~\ref{thm:main} is that one can simultaneously receive $r-1$ data streams using $r$ receive chains, while suppressing any number of external interferers.
\begin{coro}[Receiver for $r-1$ desired directions]
\label{thm:main_corollary}
Assume that we have $r-1$ desired users where $r$ is the number of receive chains over a LOS channel as defined in \eqref{eq:IC_model2_cellular}.
Then, for any rate satisfying:
\begin{equation}
R_k \le \log \left(1+\frac{P}{\gs^2} \right)
\end{equation}
there is a sufficiently large array with $N_r$ antennas and a selection matrix $\mS_{R}$, such that for all users $k=1,\ldots,r-1$, error-free transmission is achievable.

\end{coro}
\begin{proof}
The claim follows by assigning one of the antennas as a reference for interference cancellation and applying Theorem~\ref{thm:main}  to every desired direction $\gth_i$ and  choosing the second antenna with the appropriate spacing and a beamforming $\vw_{n_i}$.
\end{proof}
Denote
\beq
\label{def:W}
\mW=\left[\vw_{n_1},...,\vw_{n_{r-1}} \right]
\eeq
be an $r \times (r-1)$ beamforming matrix as constructed in Corollary \ref{thm:main_corollary}.
The following theorem characterizes the achievable sum rate for the LOS massive MU-MIMO channel.
\begin{theorem}[Main Theorem]
Choosing $\mW$ as defined in \eqref{def:W} for almost all channel realizations the achievable sum-rate satisfies
\begin{equation}
R_{\rm sum}\left(\mW\right) \ge (r-1)\log\left(1+\frac{P}{\gs^2}\right)-\gre
\end{equation}
 for every $\gre >0$.
\label{main_theorem}
\end{theorem}
\begin{proof}
let
\begin{equation}
\label{def:H}
\mH=\left[\vh((\gth_1,...\vh(\gth_K)\right]
\end{equation}
be the channel matrix including the desired and interfering signals.
Denote the effective channel matrix
\begin{equation}
    \mHt=\mW^H \mH.
\end{equation}
Let $\mQ=\mHt \mHt^H$.
By the choice of $\mW$ we obtain that
\begin{equation}
  \mQ_{i,j}=\mI+\mE
\end{equation}
where $\left|\mE_{i,j}\right| \le K \gd$, for $\gd<\frac{1}{2K}$.

Using Gershgorin's theorem \cite{horn1990matrix} we obtain that
\begin{equation}
1-K\gd\le \gm_i(\mQ)\le 1+K\gd.
\end{equation}
where $\gm_i(\mQ)$ are the eigenvalues of $\mQ$.
Hence
\begin{equation}
 R_{\min}  \le    \log \left|\mI+\frac{P}{\gs^2}\mQ  \right|\le R_{\max}
\end{equation}
where
\begin{eqnarray}
R_{\min} = (r-1)\log\left(1+\frac{P}{\gs^2}(1-K\gd) \right) \\
 R_{\max}=(r-1)\log\left(1+\frac{r}{r-1}\frac{P}{ \gs^2}\right)
\end{eqnarray}
for all $0<\gd<\frac{1}{2K}$. The second inequality follows from the MIMO capacity in the absence of external interference.
Choosing $\gd$ sufficiently small and by continuity of $\log(t)$ the claim follows.
\end{proof}

\section{Algorithms and implementation}
Theorem \ref{main_theorem}  guarantees that interference can be suppressed to any desired level. However, it does not exploit the full optimization parameter space. Ultimately, our goal is to maximize the signal to interference plus noise ratio by properly choosing the antennas and the  beamformers corresponding to each source .
The straightforward approach would be to enumerate over all subsets of $r$ antennas and evaluating the SINR, for the optimal linear receiver. The complexity of this algorithm is prohibitive and simpler algorithms are required.  For instance, we can select a single antenna as a reference antenna and choose one antenna per desired user.

 \subsection{Pairwise antenna selection}
 The simplest approach is to use Theorem \ref{main_theorem} directly where for each use we choose the antenna which maximizes the SINR  for this user.
 To that end, define
 \begin{align*}
    \vw_i&=\arg \max_{\vw\in {\mathbb C}^{r}} \frac{P|\vw^H \vh(\gth_i)|^2}
    {\sum_{j \neq i} P_j|\vw^H \vh(\gth_j)|^2+\gs^2 \|\vw\|^2} \\
    {\rm subject\ to}:& \quad  \|\vw\|_H=2 \\
                     &  \quad  w_0 =1.
\end{align*}
 where $\|\cdot \|$ is the $\ell_0$ norm.

 Each signal is obtained by beamforming the output of the reference antenna and a second antenna which provides the interference free reception.
However, since we have $r$ receivers, it is reasonable to improve the performance by optimizing the receive vector for each signal jointly, over all the selected receive chains.
To that end let $n_0,...,n_{r-1}$ be the indices of the selected antennas and let $\mS_R$ be the corresponding selection matrix defined by $\left(n_0,...,n_{r-1}\right)$.
We can now maximize the SINR of each user by optimizing
\begin{align*}
    \vw_i&=\arg \max_{\vw\in {\mathbb C}^{N_r}} \frac{P|\vw^H \mS_R \vh(\gth_i)|^2}
    {\sum_{j \neq i} P_j|\vw^H \va(\gth_j)|^2+\gs^2 \|\vw\|^2}\\
     {\rm subject\ to}:& \quad  {\rm supp}(\vw_i)=\left(n_0,...,n_{r-1}\right),
\end{align*}
where ${\rm supp}$ defines the support.
Since the directions of interferers are assumed known, we can use the interference covariance based beamformer \cite{gu2012robust}
\begin{equation}
    \vw_i=\mR_n^{-1} \mS^H_R\vh(\gth_i)
\end{equation}
where ${\rm supp}(\vw_i)$ is the support of $\vw_i$.
\begin{equation}
  \mR_n=\sum_{j \neq i} P_j \mS_R^H \vh(\gth_j)\vh(\gth_j)^H\mS_R+\gs^2 \mI.
\end{equation}
As discussed in \cite{ehrenberg2010sensitivity}, there is significant benefit in terms of robustness when using the interference covariance as a basis for beamforming instead of the received signal covariance matrix.

\section{Simulations}
To test the proposed MU-MIMO scheme for a LOS channel, we assumed that we have two desired users and 3 undesired users which serve as out of cell interference. We used 100 antennas with $\gl/2$
spacing.We performed 500 experiments. In each experiment we randomly picked the direction of all 5 signals. We assumed three receive chains ($r=3$). The first array element served as a common reference and the other elements were optimized for each desired user, respectively.

\begin{figure}[t]
\begin{center}
\includegraphics[width=0.85\columnwidth]{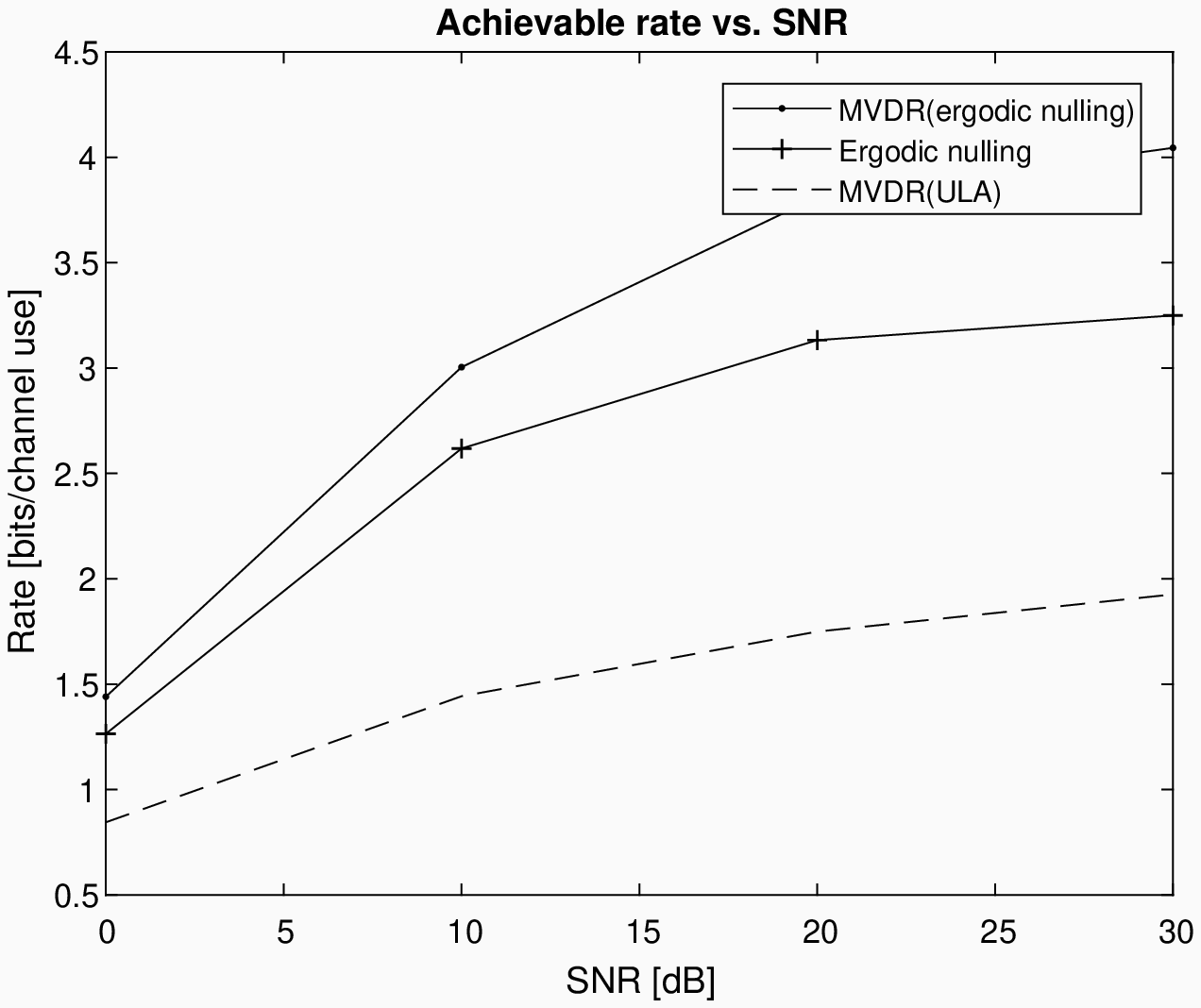}
\caption{$K=5$ users with $r-1=2$ desired users, $r=3$ receive chains. $N_r=100$.  LOS MU-MIMO channel. 500 random channels. SIR=-7dB. }
\end{center}
\label{fig:vsnr}
\end{figure}

\begin{figure}[t]
\begin{center}
\includegraphics[width=0.9\columnwidth]{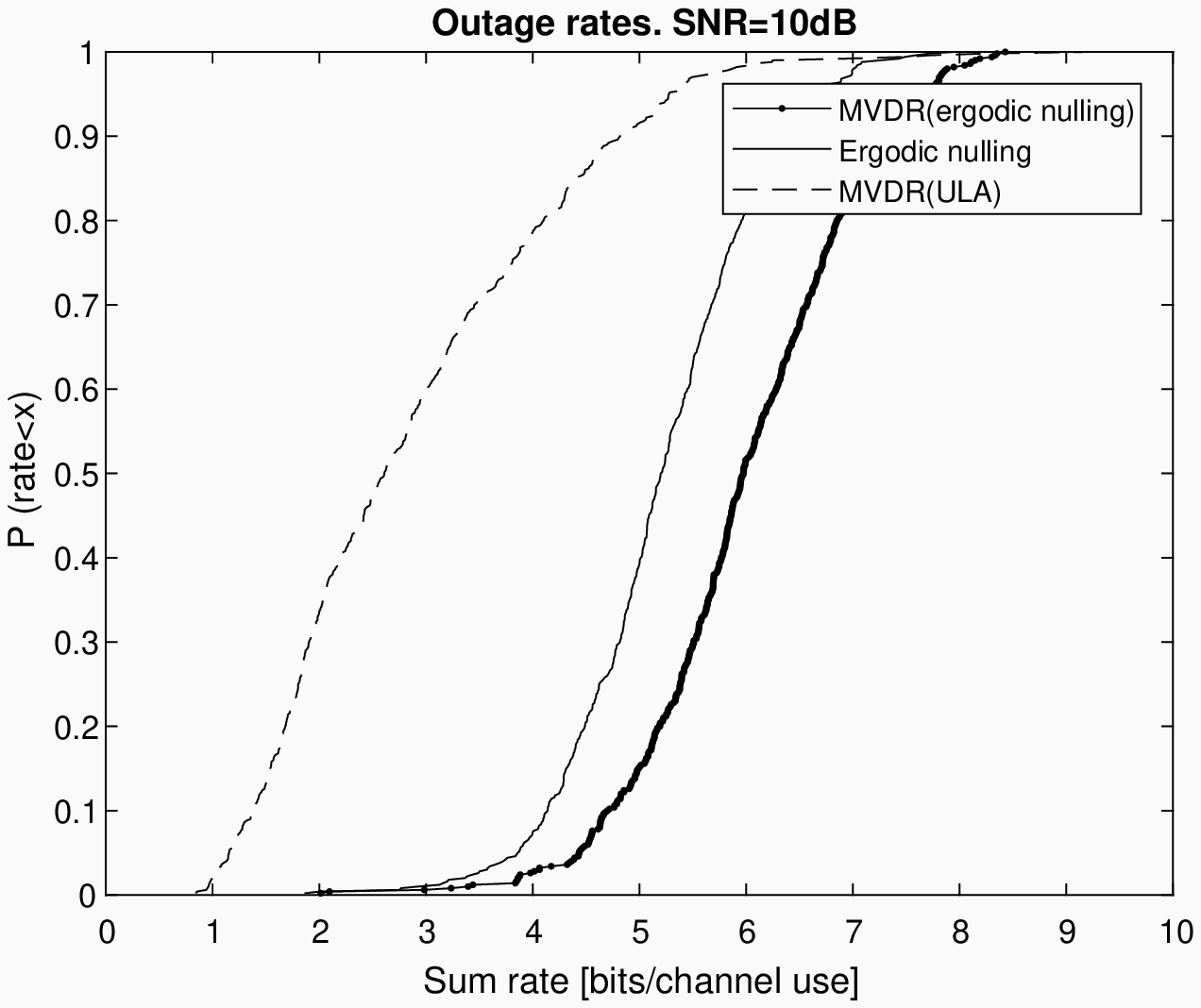}
\caption{$K=5$ users with $r-1=2$ desired users, $r=3$ receive chains. $N_r=100$.  LOS MU-MIMO channel. 500 random channels. SIR=-7dB.}
\end{center}
\label{fig:vd}
\end{figure}

\section{Extensions and Conclusions}

In this paper we proposed a novel technique for interference mitigation in LOS massive MIMO cellular networks employing a minimal number of receive chains. Simulation results demonstrate the effectiveness of the approach.

The paper focused on the uplink direction. However, by invoking a simple uplink-downlink duality argument, the same architecture applies for the downlink, assuming that directional CSI is available to the base-station. In LOS channel this assumption is mild, since practically all receiver have GPS receivers. Alternatively, the directions can be estimated using the well known techniques of \cite{sheinvald1999direction}, \cite{wax1997joint}.

\bibliographystyle{IEEEtran}
\bibliography{mybib1}

\begin{thebibliography}{10}
\providecommand{\url}[1]{#1}
\csname url@samestyle\endcsname
\providecommand{\newblock}{\relax}
\providecommand{\bibinfo}[2]{#2}
\providecommand{\BIBentrySTDinterwordspacing}{\spaceskip=0pt\relax}
\providecommand{\BIBentryALTinterwordstretchfactor}{4}
\providecommand{\BIBentryALTinterwordspacing}{\spaceskip=\fontdimen2\font plus
\BIBentryALTinterwordstretchfactor\fontdimen3\font minus
  \fontdimen4\font\relax}
\providecommand{\BIBforeignlanguage}[2]{{%
\expandafter\ifx\csname l@#1\endcsname\relax
\typeout{** WARNING: IEEEtran.bst: No hyphenation pattern has been}%
\typeout{** loaded for the language `#1'. Using the pattern for}%
\typeout{** the default language instead.}%
\else
\language=\csname l@#1\endcsname
\fi
#2}}
\providecommand{\BIBdecl}{\relax}
\BIBdecl

\bibitem{marzetta2010noncooperative}
T.~L. Marzetta, ``Noncooperative cellular wireless with unlimited numbers of
  base station antennas,'' \emph{IEEE Transactions on Wireless Communications},
  vol.~9, no.~11, pp. 3590--3600, 2010.

\bibitem{almers2007survey}
P.~Almers, E.~Bonek, A.~Burr, N.~Czink, M.~Debbah, V.~Degli-Esposti,
  H.~Hofstetter, P.~Ky{\"o}sti, D.~Laurenson, G.~Matz \emph{et~al.}, ``Survey
  of channel and radio propagation models for wireless {MIMO} systems,''
  \emph{EURASIP Journal on Wireless Communications and Networking}, vol. 2007,
  no.~1, p. 019070, 2007.

\bibitem{rappaport2013millimeter}
T.~S. Rappaport, S.~Sun, R.~Mayzus, H.~Zhao, Y.~Azar, K.~Wang, G.~N. Wong,
  J.~K. Schulz, M.~Samimi, and F.~Gutierrez~Jr, ``Millimeter wave mobile
  communications for 5g cellular: It will work!'' \emph{IEEE access}, vol.~1,
  no.~1, pp. 335--349, 2013.

\bibitem{rusek2012scaling}
F.~Rusek, D.~Persson, B.~K. Lau, E.~G. Larsson, T.~L. Marzetta, O.~Edfors, and
  F.~Tufvesson, ``Scaling up {MIMO}: Opportunities and challenges with very
  large arrays,'' \emph{arXiv preprint arXiv:1201.3210}, 2012.

\bibitem{larsson2014massive}
E.~G. Larsson, O.~Edfors, F.~Tufvesson, and T.~L. Marzetta, ``Massive {MIMO}
  for next generation wireless systems,'' \emph{IEEE Communications Magazine},
  vol.~52, no.~2, pp. 186--195, 2014.

\bibitem{liang2014low}
L.~Liang, W.~Xu, and X.~Dong, ``Low-complexity hybrid precoding in massive
  multiuser {MIMO} systems,'' \emph{IEEE Wireless Communications Letters},
  vol.~3, no.~6, pp. 653--656, 2014.

\bibitem{mollen2017uplink}
C.~Moll{\'e}n, J.~Choi, E.~G. Larsson, and R.~W. Heath, ``Uplink performance of
  wideband massive {MIMO} with one-bit adcs,'' \emph{IEEE Transactions on
  Wireless Communications}, vol.~16, no.~1, pp. 87--100, 2017.

\bibitem{mehanna2013joint}
O.~Mehanna, N.~D. Sidiropoulos, and G.~B. Giannakis, ``Joint multicast
  beamforming and antenna selection,'' \emph{IEEE Transactions on Signal
  Processing}, vol.~61, no.~10, pp. 2660--2674, 2013.

\bibitem{basar2017index}
E.~Basar, M.~Wen, R.~Mesleh, M.~Di~Renzo, Y.~Xiao, and H.~Haas, ``Index
  modulation techniques for next-generation wireless networks,'' \emph{IEEE
  Access}, vol.~5, pp. 16\,693--16\,746, 2017.

\bibitem{ishikawa201850}
N.~Ishikawa, S.~Sugiura, and L.~Hanzo, ``50 years of permutation, spatial and
  index modulation: From classic {RF} to visible light communications and data
  storage,'' \emph{IEEE Communications Surveys \& Tutorials}, 2018.

\bibitem{rappaport2015wideband}
T.~S. Rappaport, G.~R. MacCartney, M.~K. Samimi, and S.~Sun, ``Wideband
  millimeter-wave propagation measurements and channel models for future
  wireless communication system design,'' \emph{IEEE Transactions on
  Communications}, vol.~63, no.~9, pp. 3029--3056, 2015.

\bibitem{gershman2005space}
A.~Gershman and N.~Sidiropoulos, \emph{Space-time processing for {MIMO}
  communications}.\hskip 1em plus 0.5em minus 0.4em\relax John Wiley \& Sons,
  2005.

\bibitem{sadek2007active}
M.~Sadek, A.~Tarighat, and A.~H. Sayed, ``Active antenna selection in multiuser
  {MIMO} communications,'' \emph{IEEE Transactions on Signal Processing},
  vol.~55, no.~4, pp. 1498--1510, 2007.

\bibitem{gorokhov2003receive}
A.~Gorokhov, D.~A. Gore, and A.~J. Paulraj, ``Receive antenna selection for
  {MIMO} spatial multiplexing: theory and algorithms,'' \emph{IEEE Transactions
  on signal processing}, vol.~51, no.~11, pp. 2796--2807, 2003.

\bibitem{gao2013antenna}
X.~Gao, O.~Edfors, J.~Liu, and F.~Tufvesson, ``Antenna selection in measured
  massive {MIMO} channels using convex optimization,'' in \emph{2013 IEEE
  globecom workshops (GC Wkshps)}.\hskip 1em plus 0.5em minus 0.4em\relax IEEE,
  2013, pp. 129--134.

\bibitem{dong2011adaptive}
K.~Dong, N.~Prasad, X.~Wang, and S.~Zhu, ``Adaptive antenna selection and tx/rx
  beamforming for large-scale {MIMO} systems in 60 ghz channels,''
  \emph{EURASIP Journal on Wireless Communications and Networking}, vol. 2011,
  no.~1, p.~59, 2011.

\bibitem{gharavi2004fast}
M.~Gharavi-Alkhansari and A.~B. Gershman, ``Fast antenna subset selection in
  {MIMO} systems,'' \emph{IEEE transactions on signal processing}, vol.~52,
  no.~2, pp. 339--347, 2004.

\bibitem{leshem2019ergodic}
A.~Leshem and U.~Erez, ``Achieving interference free rates via ergodic
  nulling,'' arxiv. Submitted to IEEE Trans. on Information Theory.

\bibitem{weyl1916gleichverteilung}
H.~Weyl, ``{\"U}ber die {G}leichverteilung von {Z}ahlen mod. {E}ins,''
  \emph{Mathematische Annalen}, vol.~77, no.~3, pp. 313--352, 1916.

\bibitem{horn1990matrix}
R.~A. Horn, R.~A. Horn, and C.~R. Johnson, \emph{Matrix analysis}.\hskip 1em
  plus 0.5em minus 0.4em\relax Cambridge university press, 1990.

\bibitem{gu2012robust}
Y.~Gu and A.~Leshem, ``Robust adaptive beamforming based on interference
  covariance matrix reconstruction and steering vector estimation,'' \emph{IEEE
  Transactions on Signal Processing}, vol.~60, no.~7, pp. 3881--3885, 2012.

\bibitem{ehrenberg2010sensitivity}
L.~Ehrenberg, S.~Gannot, A.~Leshem, and E.~Zehavi, ``Sensitivity analysis of
  mvdr and mpdr beamformers,'' in \emph{Electrical and Electronics Engineers in
  Israel (IEEEI), 2010 IEEE 26th Convention of}.\hskip 1em plus 0.5em minus
  0.4em\relax IEEE, 2010, pp. 000\,416--000\,420.

\bibitem{sheinvald1999direction}
J.~Sheinvald and M.~Wax, ``Direction finding with fewer receivers via
  time-varying preprocessing,'' \emph{IEEE transactions on signal processing},
  vol.~47, no.~1, pp. 2--9, 1999.

\bibitem{wax1997joint}
M.~Wax and A.~Leshem, ``Joint estimation of time delays and directions of
  arrival of multiple reflections of a known signal,'' \emph{IEEE Transactions
  on Signal Processing}, vol.~45, no.~10, pp. 2477--2484, 1997.

\end{thebibliography}



\end{document}